\documentclass[11pt, onecolumn]{article}

\usepackage{ifxetex}
\ifxetex
\usepackage{mathspec}
\usepackage{fontspec}
\else
\usepackage[utf8]{inputenc}
\usepackage[T1]{fontenc}
\fi

\usepackage[margin=1in,a4paper]{geometry}

\usepackage{graphicx}
\usepackage{url}
\usepackage{amsmath,amssymb,amsthm}
\usepackage{tikz}
\usetikzlibrary{arrows,automata,matrix}
\usepackage{array}

\usepackage{listings}

\usepackage[style=numeric-comp,backend=biber]{biblatex}
\newcommand\Ci{C_\otimes}
\newcommand\Co{C_\odot}
\newcommand\Li{L_\otimes}
\newcommand\Lo{L_\odot}
\newcommand\Ri{R_\otimes}
\newcommand\Ro{R_\odot}

\newcommand\Tie{\mathcal{T}\!\textit{ie}}
\newcommand\bnfrule[1]{\langle\textit{#1}\rangle}

\newtheorem{theorem}{Theorem}

\makeatletter
\g@addto@macro\maketitle{\thispagestyle{empty}\setcounter{page}{0}}
\makeatother

\begin{document}

\title{More ties than we thought}
\author{
\textsc{Dan Hirsch} \\ Upstanding Hackers Inc. \\ \texttt{thequux@upstandinghackers.com}
\\[1cm]
\textsc{Ingemar Markström} \\ 
KTH Royal Institute of Technology, Stockholm \\
\texttt{ingemarm@kth.se} \\[1cm]
\textsc{Meredith L. Patterson} \\ Upstanding Hackers Inc. \\ \texttt{mlp@upstandinghackers.com}
\\[1cm]
\textsc{Anders Sandberg} \\ Oxford University \\ \texttt{anders.sandberg@philosophy.ox.ac.uk}
\\[1cm]
\textsc{Mikael Vejdemo-Johansson} \\ KTH Royal Institute of Technology, Stockholm \\ Jo\v zef \v Stefan Institute, Ljubljana \\ Institute for Mathematics and its Applications, Minneapolis \\ \texttt{mvj@kth.se}
\\[1cm]
}

\maketitle

\begin{abstract}
We extend the existing enumeration of neck tie-knots to include tie-knots with a textured front, tied with the narrow end of a tie. These tie-knots have gained popularity in recent years, based on reconstructions of a costume detail from The Matrix Reloaded, and are explicitly ruled out in the enumeration by Fink and Mao (2000).

We show that the relaxed tie-knot description language that comprehensively describes these extended tie-knot classes is context free. It has a regular sub-language that covers all the knots that originally inspired the work.

From the full language, we enumerate 266\,682 distinct tie-knots that seem tie-able with a normal neck-tie. Out of these 266\,682, we also enumerate 24\,882 tie-knots that belong to the regular sub-language.
\end{abstract}

\newpage

\section{Introduction}
\label{sec:introduction}

\begin{figure}
\includegraphics[width=0.3\textwidth]{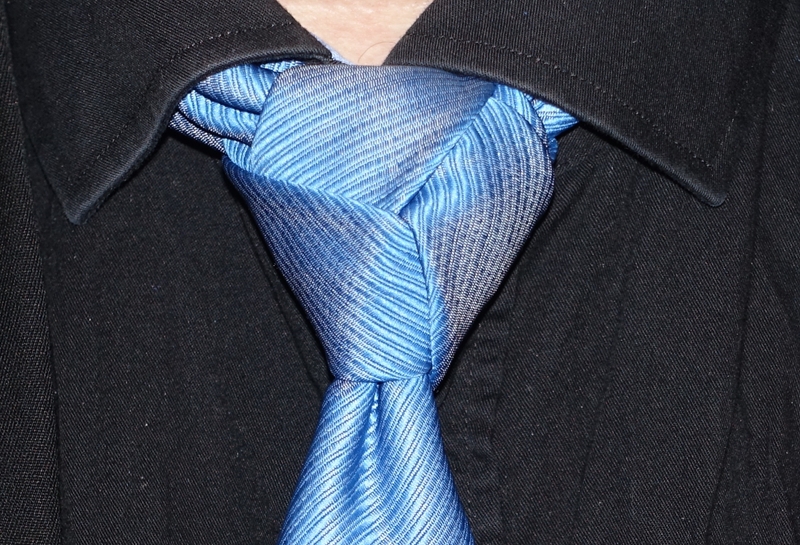} \hspace{12pt}
\includegraphics[width=0.3\textwidth]{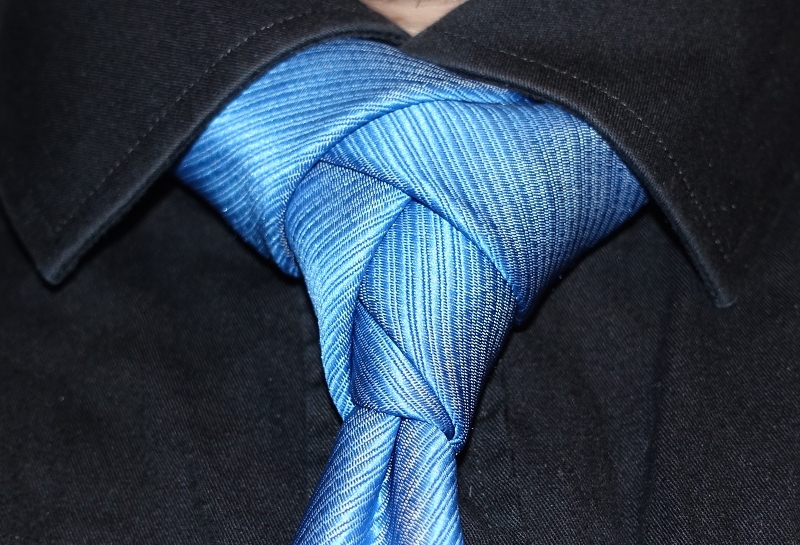} \hspace{12pt}
\includegraphics[width=0.3\textwidth]{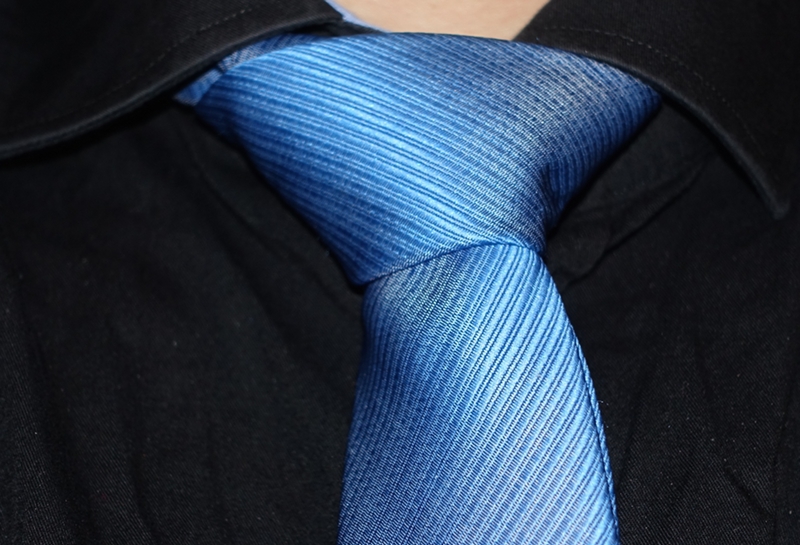} \\[12pt]
\includegraphics[width=0.3\textwidth]{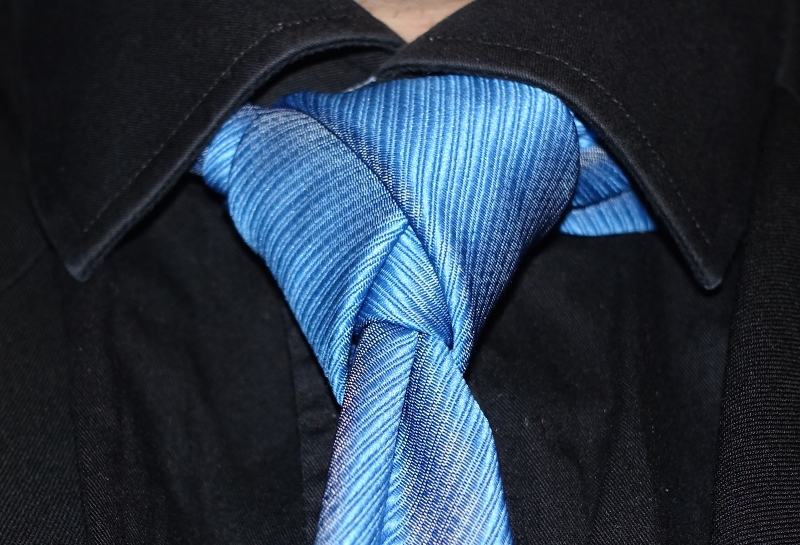} \hspace{12pt}
\includegraphics[width=0.3\textwidth]{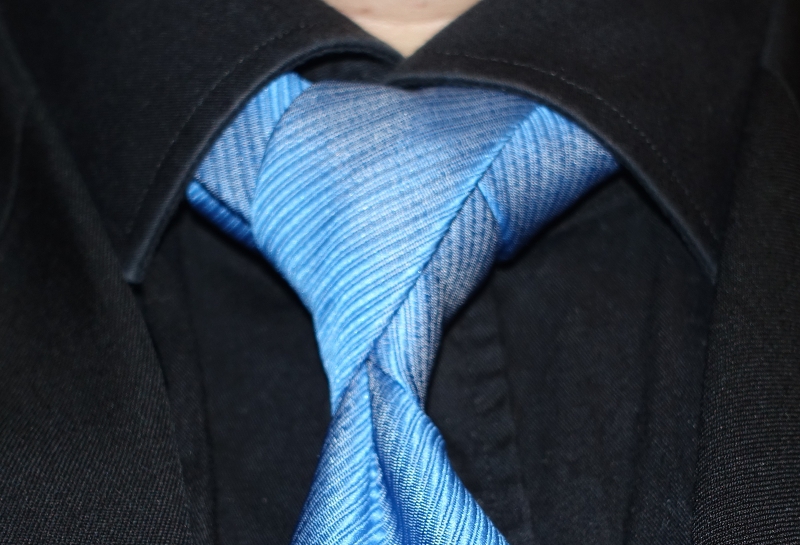} \hspace{12pt}
\includegraphics[width=0.3\textwidth]{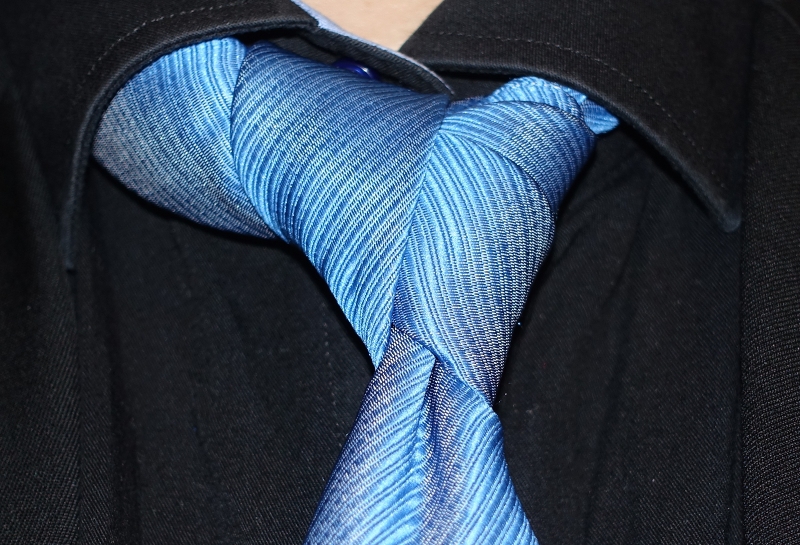} 
\caption{Some specific tie-knot examples. Top row from left: the Trinity (L-110.4), the Eldredge (L-373.2) and the Balthus (C-63.0, the largest knot listed by Fink and Mao). Bottom row randomly drawn knots. From left: L-81.0, L-625.0, R-353.0}
\end{figure}

There are several different ways to tie a necktie. Classically, knots such as the four-in-hand, the half windsor and the full windsor have commonly been taught to new tie-wearers. In a sequence of papers and a book, Fink and Mao \cite{fink200185,fink2000tie,fink1999designing} defined a formal language for describing tie-knots, encoding the topology and geometry of the knot tying process into the formal language, and then used this language to enumerate all tie-knots that could reasonably be tied with a normal-sized necktie.

The enumeration of Fink and Mao crucially depends on dictating a particular finishing sequence for tie-knots: a finishing sequence that forces the front of the knot -- the façade -- to be a flat stretch of fabric. With this assumption in place, Fink and Mao produce a list of 85 distinct tie-knots, and determine several novel knots that extend the previously commonly known list of tie-knots. 

In recent years, however, interest has been growing for a new approach to tie-knots. In \emph{The Matrix Reloaded} \cite{wachowski2003matrix}, the character of ``The Merovingian'' has a sequence of particularly fancy tie-knots. Attempts by fans of the movie to recreate the tie-knots from the Merovingian have led to a collection of new tie-knot inventions, all of which rely on tying the tie with the thin end of the tie -- the thin blade. Doing this allows for a knot with textures or stylings of the front of the knot, producing symmetric and pleasing patterns.

\textcite{eldredge-reloaded} gives the history of the development of novel tie-knots. It starts out in 2003 when the \emph{edeity knot} is published as a PDF tutorial. Over the subsequent 7 years, more and more enthusiasts involve themselves, publish new approximations of the Merovingian tie-knot as PDF files or YouTube videos. By 2009, the new tie-knots are featured on the website Lifehacker and go viral.

In this paper, we present a radical simplification of the formal language proposed by Fink and Mao, together with an analysis of the asymptotic complexity class of the tie-knots language. We produce a novel enumeration of necktie-knots tied with the thin blade, and compare it to the results of Fink and Mao.

\subsection{Formal languages}
\label{sec:formal-languages}

The work in this paper relies heavily on the language of formal languages, as used in theoretical computer science and in mathematical linguistics. For a comprehensive reference, we recommend the textbook by \textcite{sipser2006introduction}.

Recall that given a finite set $\mathcal L$ called an \emph{alphabet}, the set of all sequences of any length of items drawn (with replacement) from $\mathcal L$ is denoted by $\mathcal L^*$. A \emph{formal language} on the alphabet $\mathcal L$ is some subset $\mathcal A$ of $\mathcal L^*$. The complexity of the automaton required to determine whether a sequence is an element of $\mathcal A$ places $\mathcal A$ in one of several complexity classes. Languages that are described by finite state automata are \emph{regular}; languages that require a pushdown automaton are \emph{context free}; languages that require a linear bounded automaton are \emph{context sensitive} and languages that require a full Turing machine to determine are called \emph{recursively enumerable}. This sequence builds an increasing hierarchy of expressibility and computational complexity for syntactic rules for strings of some arbitrary sort of tokens.

One way to describe a language is to give a \emph{grammar} -- a set of production rules that decompose some form of abstract tokens into sequences of abstract or concrete tokens, ending with a sequence of elements in some alphabet. The standard notation for such grammars is the Backus-Naur form, which uses $::=$ to denote the production rules and $\bnfrule{some name}$ to denote the abstract tokens. Further common symbols are $*$, the Kleene star, that denotes an arbitrary number of repetitions of the previous token (or group in brackets), and $|$, denoting a choice of one of the adjoining options.

\section{The anatomy of a necktie}
\label{sec:anatomy-necktie}

In the following, we will often refer to various parts and constructions with a necktie. We call the ends of a necktie \emph{blades}, and distinguish between the \emph{broad blade} and the \emph{thin blade}\footnote{There are neckties without a width difference between the ends. We ignore this distinction for this paper.} -- see Figure~\ref{fig:labels} for these names. The tie-knot can be divided up into a \emph{body}, consisting of all the twists and turns that are not directly visible in the final knot, and a \emph{façade}, consisting of the parts of the tie actually visible in the end. In Figure~\ref{fig:bodyfacade} we demonstrate this distinction. The body builds up the overall shape of the tie-knot, while the façade gives texture to the front of the knot. The enumeration of Fink and Mao only considers knots with trivial façades, while these later inventions all consider more interesting façades. As a knot is in place around a wearer, the Y-shape of the tie divides the torso into 3 regions: Left, Center and Right -- as shown to the right in Figure~\ref{fig:labels}.

\begin{figure}[t]
  \centering
  \includegraphics{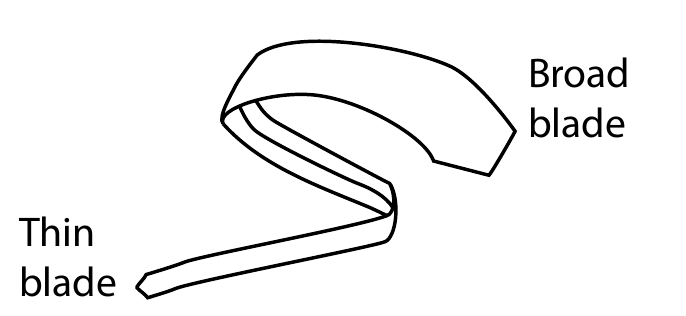} \qquad 
  \includegraphics{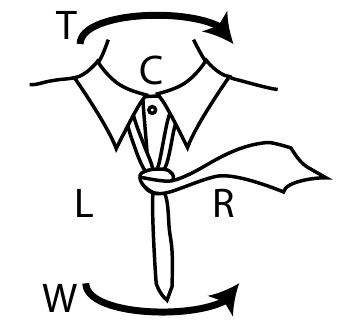}
  \caption{The parts of a necktie, and the division of the wearer's torso with the regions (Left, Center Right) and the winding directions (Turnwise, Widdershins) marked out for reference.}
  \label{fig:labels}
\end{figure}

A tie-knot has to be tied by winding and tucking one of the two blades around the other: if both blades are active, then the tie can no longer be adjusted in place for a comfortable fit. We shall refer to the blade used in tying the knot as the \emph{leading blade} or the \emph{active blade}. Each time the active blade is moved across the tie-knot -- in front or in back -- we call the part of the tie laid on top of the knot a \emph{bow}. 

\begin{figure}[t]
  \centering
  \includegraphics{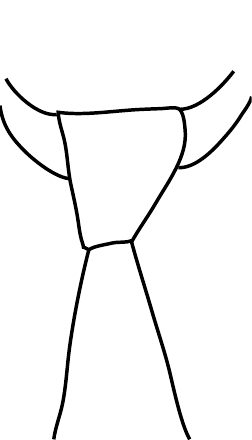} \qquad
  \includegraphics{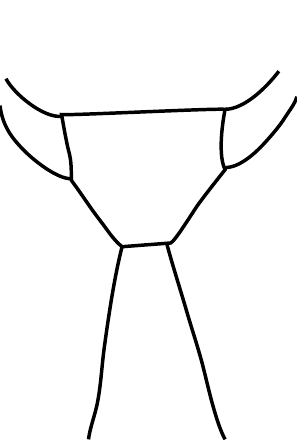} \qquad
  \includegraphics{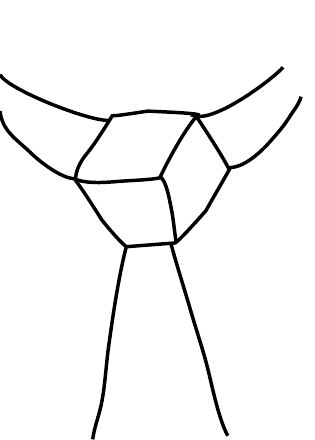} 
  \caption{Different examples of tie-knots. Left, a 4-in-hand; middle, a
double windsor; right a trinity. The 4-in-hand and double windsor
share the flat façade but have different bodies producing different
shapes. The trinity has a completely different façade, produced by a
different wind and tuck pattern. }
\label{fig:bodyfacade}
\end{figure}

\section{A language for tie-knots}
\label{sec:language-tie-knots}

\textcite{fink2000tie} observe that once the first crossing has been made, the wrapping sequence of a classical tie-knot is completely decided by the sequence of regions into which the broad blade is moved. Adorning the region specifications with a direction -- is the tie moving away from the wearer or towards the wearer -- they establish a formal alphabet for describing tie-knots with 7 symbols. We reproduce their construction here, using $U$ for the move to tuck the blade \textit{U}nder the tie itself\footnote{\citeauthor{fink2000tie} used $T$ for \textit{T}uck}.
The notation proposed by \textcite{fink2000tie} interprets repetitions $U^k$ of $U$ as tucking the blade $k$ bows under the top. It turns out that the complexity analysis is far simpler if we instead write $U^k$ for tucking the blade under the bow that was produced $2k$ windings ago. This produces a language on the alphabet:
\[
\{\Li, \Lo, \Ci, \Co, \Ri, \Ro, U\}
\]

They then introduce relations and restrictions on these symbols:
\begin{description}
\item[$\Tie_1$] No region $(L, C, R)$ shall repeat: after an $L$ only $C$ or $R$ are valid next regions. $U$ moves do not influence this.
\item[$\Tie_2$] No direction ($\odot$ -- out of the paper, $\otimes$ -- in towards the paper) shall repeat: after an outwards move, the next one must go inwards. $U$ moves do not influence this.
\item[$\Tie_3$] Tucks ($U$) are valid after an outward move.
\item[$\Tie_4$] A tie-knot can end only on one of $\Ci, \Co$ or $U$. In fact, almost all classical knots end on $U$.\footnote{The exemption here being the Onassis style knot, favored by the eponymous shipping magnate, where after a classical knot the broad blade is brought up with a $\Co$ move to fall in front of the knot, hiding the knot completely.}
\item[$\Tie_5$] A $k$-fold tuck $U^k$ is only valid after at least $2k$ preceding moves. \textcite{fink2000tie} do not pay much attention to the conditions on $k$-fold tucks, since these show up in their enumeration as stylistic variations, exclusively at the end of a knot.
\end{description}

This collection of rules allow us to drastically shrink the tie language, both in alphabet and axioms. \citeauthor{fink1999designing} are careful to annotate whether tie-knot moves go outwards or inwards at any given point. We note that the inwards/outwards distinction follows as a direct consequence of axioms $\Tie_2$, $\Tie_3$ and $\Tie_4$. Since non-tuck moves must alternate between inwards and outwards, and the last non-tuck move must be outwards, the orientation of any sequence of moves follows by backtracking from the end of the string.

Hence, when faced with a non-annotated string like 
\[
\texttt{RCLCRCLCRCLRURCLU}
\]
we can immediately trace from the tail of the knot string: the last move before the final tuck must be outwards, so that \texttt{L} must be a $\Lo$. So it must be preceded by $\Ro\Ci$. Tracing backwards, we can specify the entire string above to
\[
\Ri\Co\Li\Co\Ri\Co\Li\Co\Ri\Co\Li\Ro U\Ci\Ro\Ci\Lo U
\]

Next, the axiom $\Tie_1$ means that a sequence will not contain either of $\texttt{LU}^*\texttt{L}, \texttt{CU}^*\texttt{C}, \texttt{RU}^*\texttt{R}$ as subsequences\footnote{Recall that the Kleene star $F^*$ is used to denote sequences of 0 or more repetitions of the string $F$.}. Hence, the listing of regions is less important than the direction of transition: any valid transition is going to go either clockwise or counterclockwise\footnote{Say, as seen on the mirror image. Changing this convention does not change the count, as long as the change is consequently done.}. Writing \texttt{T} for clockwise\footnote{\texttt{T} for Turnwise} and \texttt{W} for counterclockwise\footnote{\texttt{W} for Widdershins}, we can give a strongly reduced tie language on the alphabet \texttt{T, W, U}. To completely determine a tie-knot, the sequence needs a starting state: an annotation on whether the first crossing of a tie-knot goes across to the right or to the left. In such a sequence, a \texttt{U} instruction must be followed by either \texttt{T} or \texttt{W} dictating which direction the winding continues after the tuck, unless it is the last move of the tie: in this case, the blade is assumed to continue straight ahead -- down in front for most broad-blade tie-knots, tucked in under the collar for most thin-blade knots. 

The position of the leading blade after a sequence of \texttt{W}/\texttt{T} windings is a direct result of $\#W - \#T \pmod{3}$. This observation allows us to gain control over several conditions determining whether a distribution of \texttt{U} symbols over a sequence of \texttt{W}/\texttt{T} produces a physically viable tie-knot. 

\begin{theorem}
  A position in a winding sequence is valid for a $k$-fold tuck if the sub-sequence of the last $2k$ \texttt{W} or \texttt{T} symbols is such that either
  \begin{enumerate}
  \item starts with \texttt{W} and satisfies $\#W-\#T = 2 \pmod 3$
  \item starts with \texttt{T} and satisfies $\#T-\#W = 2 \pmod 3$
  \end{enumerate}
\end{theorem}

\begin{proof}
  The initial symbol produces the bow under which the tuck will go. If the initial symbol goes, say, from \texttt{R} to \texttt{L}, then the tuck move needs to come from \texttt{C} in order to go under the bow. In general, a tuck needs to come from the one region not involved in the covering bow. Every other bow goes in front of the knot, and the others go behind the knot. Hence, there are $2k-1$ additional winding symbols until the active blade returns to the right side of the knot. During these $2k-1$ symbols, we need to transition one more step around the sequence of regions. The transitions \texttt{W} and \texttt{T} are generator and inverse for the cyclic group of order 3, concluding the proof.
\end{proof}

It is worth noticing here that  a particular point along a winding can be simultaneously valid for both a $k$-fold and an $m$-fold tuck for $k\neq m$. One example would be in the winding string \texttt{TWTT}: ending with \texttt{TT}, it is a valid site for a 1-fold tuck producing \texttt{TWTTU}, and since \texttt{TWTT} starts with \texttt{T} and has 2 more \texttt{T} than \texttt{W}, it is also a valid site for a 2-fold tuck producing \texttt{TWTTUU}. We will revisit this example below, in Section~\ref{sec:depthN}.

We may notice that with the usual physical constraints on a tie -- where we have experimentally established that broad blade ties tend to be bounded by 9 moves, and thin blade ties by 15 moves, we can expect that no meaningful tuck deeper than 7 will ever be relevant; 4 for the broad blade ties. The bound of 4 is achieved in the enumeration by \textcite{fink1999designing}.

In our enumeration, we will for the sake of comfort focus on ties up to 13 moves.

\section{Language complexity}
\label{sec:language-complexity}

In this section, we examine the complexity features of the tie-knot language. Due to the constraints we have already observed on the cardinality of \texttt{W} and \texttt{T}, we will define a grammar for this language. We will write this grammar with a Backus-Naur form. Although in practice it is only possible to realise finite strings in the tie-knot language due to the physical properties of fabric, we assume an arbitrarily long (but finite), infinitely thin tie.

\subsection{Single-depth tucks}
\label{sec:depth1}

The classical Fink and Mao system has a regular grammar, given by

\begin{align*}
\bnfrule{tie} &::= \texttt{L}\bnfrule{L} \\
\bnfrule{lastR} &::= \texttt{L}\bnfrule{lastL} \,|\, \texttt{C}\bnfrule{lastC} \,|\, \texttt{LCU}\\
\bnfrule{lastL} &::= \texttt{R}\bnfrule{lastR} \,|\, \texttt{C}\bnfrule{lastC} \,|\, \texttt{RCU} \\
\bnfrule{lastC} &::= \texttt{L}\bnfrule{lastL} \,|\, \texttt{R}\bnfrule{lastR} \\
\end{align*}

We use the symbol $\bnfrule{lastR}$ to denote the rule that describes what can happen when the last move seen was an \texttt{R}. 
Hence, at any step in the grammar, some tie knot symbol is emitted, and the grammar continues from the state that that symbol was the last symbol emitted.

The above grammar works well if the only tucks to appear are at the end. For intermediate tucks, and to avoid tucks to be placed at the back of the knot (obeying $\Tie_3$), we would need to keep track of winding parity: tucks are only valid an even number of winding steps from the end.

We can describe this with a regular grammar. For the full tie-knot language, the grammar will end up context-free, as we will see in Section~\ref{sec:depthN}.

\begin{align*}
\bnfrule{tie} &::= \bnfrule{prefix} (\bnfrule{pair} \,|\, \bnfrule{tuck})* \bnfrule{tuck} \\
\bnfrule{prefix} &::= \texttt{T} \,|\, \texttt{W} \,|\, \epsilon \\
\bnfrule{pair} &::= (\texttt{T}|\texttt{W})(\texttt{T}|\texttt{W}) \\
\bnfrule{tuck} &::= \texttt{TTU} \,|\, \texttt{WWU}
\end{align*}

The distribution of \texttt{T} and \texttt{W} varies by type of knot: for classical knots, $\#W-\#T = 2\pmod 3$; for modern knots that tuck to the right, $\#W-\#T = 1\pmod 3$; and for modern knots that tuck to the left, $\#W-\#T = 0\pmod 3$. This grammar does not discriminate between these three sub-classes.
In order to track these sub-classes, the \texttt{RLC}-notation is easier. 

In order to rebuild this grammar to one based on the \texttt{RLC}-notation, note that from \texttt{L} a \texttt{T} takes us to \texttt{C} and a \texttt{W} takes us to \texttt{R}. So from a $\bnfrule{lastT}$ residing at \texttt{L}, we have the options: \texttt{W} to \texttt{R}, \texttt{T} to \texttt{C}, or \texttt{TU} to \texttt{C}. In particular, there is a $\bnfrule{lastT}$ at \texttt{L} if we arrived from \texttt{R}. Hence, the \texttt{TU} option can be seen as being a \texttt{TTU} option executing from the preceding \texttt{R} state.

There is thus, at any given position in the tie sequence, the options of proceeding with a \texttt{T} or a \texttt{W}, or to proceed with one of \texttt{TTU} or \texttt{WWU}. In the latter two cases, we can also accept the string. Starting at \texttt{L}, these options take us -- in order -- to \texttt{C}, to \texttt{R}, to \texttt{CRU} and to \texttt{RCU} respectively. This observation extends by symmetry to all stages, giving the grammar below.

\begin{align*}
\bnfrule{lastR} &::= \texttt{LR}\bnfrule{lastR} \,|\, \texttt{CR}\bnfrule{lastR} \,|\, 
  \texttt{LC}\bnfrule{lastC} \,|\, \texttt{CL}\bnfrule{lastL} \,|\, \\
  &\phantom{::=} \texttt{LCU} [\bnfrule{lastC}] \,|\,
  \texttt{CLU} [\bnfrule{lastL}] \\
\bnfrule{lastL} &::= \texttt{RL}\bnfrule{lastL} \,|\, \texttt{CL}\bnfrule{lastL} \,|\, 
  \texttt{RC}\bnfrule{lastC} \,|\, \texttt{CR}\bnfrule{lastR} \,|\, \\
  &\phantom{::=} \texttt{RCU} [\bnfrule{lastC}] \,|\,
  \texttt{CRU} [\bnfrule{lastR}] \\
\bnfrule{lastC} &::= \texttt{LC}\bnfrule{lastC} \,|\, \texttt{RC}\bnfrule{lastC} \,|\, 
  \texttt{LR}\bnfrule{lastR} \,|\, \texttt{RL}\bnfrule{lastL} \,|\, \\
  &\phantom{::=} \texttt{LRU} [\bnfrule{lastR}] \,|\,
  \texttt{RLU} [\bnfrule{lastL}] \\
\bnfrule{tie} &::= \texttt{L} (\bnfrule{lastL} \,|\, R\bnfrule{lastR} \,|\, C\bnfrule{lastC}) \\
\end{align*}

By excluding some the exit rules, this allows us to enumerate novel tie-knots with a specific ending direction, which will be of interest later on.

\subsection{Recursive tucks}
\label{sec:depthN}

We can write a context-free grammar for the arbitrary depth tuck tie-knots.

\begin{align*}
\bnfrule{tie} &::= \bnfrule{prefix} (\bnfrule{pair} \,|\, \bnfrule{tuck})* \bnfrule{tuck} \\
\bnfrule{prefix} &::= \texttt{T} \,|\, \texttt{W} \,|\, \epsilon \\
\bnfrule{pair} &::= (\texttt{T}|\texttt{W})(\texttt{T}|\texttt{W}) \\
\bnfrule{tuck} &::= \bnfrule{ttuck2} \,|\, \bnfrule{wtuck2} \\
\bnfrule{ttuck2} &::= \texttt{TT}\bnfrule{w0}\texttt{U} \,|\, \texttt{TW}\bnfrule{w1}\texttt{U} \\
\bnfrule{wtuck2} &::= \texttt{WW}\bnfrule{w0}\texttt{U} \,|\, \texttt{WT}\bnfrule{w2}\texttt{U} \\
\bnfrule{w0} &::= \texttt{WW}\bnfrule{w1}\texttt{U} \,|\, \texttt{WT}\bnfrule{w0}\texttt{U} \,|\, 
                  \texttt{TW}\bnfrule{w0}\texttt{U} |\texttt{TT}\bnfrule{w2}\texttt{U} \,|\, \\ 
    &\phantom{::=}\bnfrule{ttuck2}\texttt{'}\bnfrule{w2}\texttt{U} \,|\, 
                  \bnfrule{wtuck2}\texttt{'}\bnfrule{w1}\texttt{U} \,|\, \epsilon \\
\bnfrule{w1} &::= \texttt{WW}\bnfrule{w2}\texttt{U} \,|\, \texttt{WT}\bnfrule{w1}\texttt{U} \,|\, 
                  \texttt{TW}\bnfrule{w1}\texttt{U} |\texttt{TT}\bnfrule{w0}\texttt{U} \,|\, \\ 
    &\phantom{::=}\bnfrule{ttuck2}\texttt{'}\bnfrule{w0}\texttt{U} \,|\, 
                  \bnfrule{wtuck2}\texttt{'}\bnfrule{w2}\texttt{U}\\
\bnfrule{w2} &::= \texttt{WW}\bnfrule{w0}\texttt{U} \,|\, \texttt{WT}\bnfrule{w2}\texttt{U} \,|\, 
                  \texttt{TW}\bnfrule{w2}\texttt{U} |\texttt{TT}\bnfrule{w1}\texttt{U} \,|\, \\ 
    &\phantom{::=}\bnfrule{ttuck2}\texttt{'}\bnfrule{w1}\texttt{U} \,|\, 
                  \bnfrule{wtuck2}\texttt{'}\bnfrule{w0}\texttt{U}\\
\end{align*}

Note that the validity of a tuck depends only on the count of \texttt{T} and \texttt{W} in the entire sequence comprising the tuck, and not the validity of any tucks recursively embedded into it. For instance, \texttt{TWTT} is a valid depth-2-tuckable sequence, as is its embedded depth-1-tuckable sequence \texttt{TT}. However, \texttt{TTWT} is also a valid depth-2-tuckable sequence, even though \texttt{WT} is not a valid depth-1-tuckable sequence.

We introduce the symbol \texttt{'} to delineate different tucks that may come in immediate sequence, such as happens in the tie knot \texttt{TWTTU'UU}.

\subsection{Classification of the tie-knot language}
\label{classification}

If we limit our attention to only the single-depth tie-knots described in Section~\ref{sec:depth1}, then the grammar is regular, proving that this tie language is a regular language and can be described by a finite automaton. In particular this implies that the tie-knot language proposed by Fink and Mao \cite{fink1999designing} is regular. In fact, an automaton accepting these tie-knots is given by:


\begin{center}
\begin{tikzpicture}[->, >=stealth', 
shorten >=1pt, auto, node distance=2.8cm, semithick]
\node [initial,state] (start) {};
\node [state, right of=start] (cycle) {};
\node [accepting, state, right of=cycle] (end) {};
\node [state, above of=cycle] (pair) {};
\node [state, below of=cycle] (tuck) {};
\draw (start) edge [bend right,swap] node [near start]{\texttt{T}} (cycle);
\draw (start) edge node [near start] {$\epsilon$} (cycle);
\draw (start) edge [bend left] node [near start] {\texttt{W}} (cycle);
\draw (cycle) edge [bend left] node {\texttt{TT}} (tuck);
\draw (cycle) edge [bend right,swap] node {\texttt{WW}} (tuck);
\draw (tuck) edge node {\texttt{U}} (cycle);
\draw (cycle) edge [bend left] node {\texttt{TTU}} (end);
\draw (cycle) edge [bend right] node {\texttt{WWU}} (end);
\draw (cycle) edge [bend left=10] node [anchor=center,fill=white,near start]{\texttt{T}} (pair);
\draw (cycle) edge [bend left=50] node [anchor=center,fill=white,near start]{\texttt{W}} (pair);
\draw (pair) edge [bend left=10] node [anchor=center,fill=white,near start] {\texttt{T}} (cycle);
\draw (pair) edge [bend left=50] node [anchor=center,fill=white,near start ]{\texttt{W}} (cycle);
\end{tikzpicture}
\end{center}

After the prefix, execution originates at the middle node, but has to go outside and return before the machine will accept input. This maintains the even length conditions required by $\Tie_3$.

As for the deeper tucked language in Section~\ref{sec:depthN}, the grammar we gave shows it to be at most context-free. Whether it is exactly context-free requires us to exclude the existence of a regular grammar.

\begin{theorem}
  The deeper tucked language is context-free.
\end{theorem}
\begin{proof}
  Our grammar in Section~ \ref{sec:depthN} already shows that the language for deeper tucked tie-knots is 
  either regular or context-free: it produces tie-knot strings with only single non-terminal symbols to 
  the left of each production rule. 

  It remains to show that this language cannot be regular. To do this, we use the pumping 
  lemma for regular languages. Recall that the pumping lemma states that for every regular
  language there is a constant $p$ such that for any word $w$ of length at least $p$, there 
  is a decomposition $w=xyz$ such that $|xy| \leq p$, $|y|\geq 1$ and $xy^iz$ is a valid 
  string for all $i>0$.

  Since the reverse of any regular language is also regular, the pumping lemma has an 
  alternative statement that requires $|yz|\leq p$ instead. We shall be using this next.

  Suppose there is such a $p$. Consider the tie-knot $\texttt{TTW}^{6q-2}\texttt{U}^{3q}$ for some $q>p/3$.
  Any decomposition such that $|yz|\leq p$ will be such that $y$ and $z$ consist of only $\texttt{U}$ 
  symbols. In particular $y$ consists of only $\texttt{U}$ symbols. Hence, for sufficiently large values 
  of $i$, there are too few preceding \texttt{T}/\texttt{W}-symbols to admit that tuck depth.

  Hence the language is not regular.
\end{proof}

\section{Enumeration}
\label{sec:enumeration}

We can cut down the enumeration work by using some apparent symmetries. Without loss of generality, we can assume that a tie-knot starts by putting the active blade in region \texttt{R}: any knot starting in the region \texttt{L} is the mirror image of a knot that starts in \texttt{R} and swaps all \texttt{W} to \texttt{T} and vice versa.

\subsection{Generating functions}
\label{sec:generatingfunctions}

Generating functions have proven a powerful method for enumerative combinatorics. One very good overview of the field is provided by the textbooks by Stanley~\cite{stanley1997enumerative,stanley1999enumerative}. 
Their relevance to formal languages is based on a paper by Chomsky and Schützenberg~\cite{chomsky1959algebraic} that studied context-free grammars using formal power series. More details will appear in the (yet unpublished) Handbook AutoMathA~\cite{gruber2012enumerating}.

A generating function for a series $a_n$ of numbers is a formal power series $A(z) = \sum_{j=0}^\infty a_jz^j$ such that the coefficient of the degree $k$ term is precisely $a_k$. Where $a_k$ and $b_k$ are counts of ``things of size $k$'' of type $a$ or $b$ respectively, the sum of the corresponding generating functions is the count of ``things of size $k$'' across both categories. If gluing some thing of type $a$ with size $j$ to some thing of type $b$ with size $k$ produces a thing of size $j+k$, then the product of the generating functions measures the counts of things you get by gluing things together between the two types.

For our necktie-knot grammars, the sizes are the winding lengths of the ties, and it is clearly the case that adding a new symbol extends the size (thus is a multiplication action), and taking either one or another rule extends the items considered (thus is an additive action). 

The \texttt{Maple}\footnote{Maple is a trademark of Waterloo Maple Inc. The computations of generating functions in this paper were performed by using Maple.} package \texttt{combstruct} has built-in functions for computing a generating function from a grammar specification. Using this, and the grammars we state in Section~\ref{sec:depth1}, we are able to compute generating functions for both the winding counts and the necktie counts for both Fink and Mao's setting and the single-depth tuck setting.

\begin{itemize}
\item The generating function for Fink and Mao necktie-knots is
\[
\frac{z^3}{(1+z)(1-2z)} = z^3 + z^4 + 3z^5 + 5z^6 + 11z^7 + 21z^8 + 43z^9 + O(z^{10})
\]
\item The generating function for single tuck necktie-knots is
\begin{multline*}
\frac{2z^3(2z+1)}{1-6z^2} = 2z^3 + 4z^4 + 12z^5 + 24z^6 + 72z^7 + 144z^8 + 432z^9 + \\ 864z^{10} + 2\,592z^{11} + 
5\,184z^{12} + 15\,552z^{13}+O(z^{14})
\end{multline*}
\item By removing final states from the BNF grammar, we can compute corresponding generating functions for each of the final tuck destinations. \\
	For an \texttt{R}-final tuck, we remove all final states except for \texttt{CRU} and \texttt{LRU}, making the non-terminal symbol mandatory for all other tuck sequences. For \texttt{L}, we remove all but \texttt{CLU} and \texttt{RLU}. For \texttt{C}, we remove all but \texttt{RCU} and \texttt{LCU}. This results in the following generating functions for \texttt{R}-final, \texttt{L}-final and \texttt{C}-final sequences, respectively.
    \begin{multline*}
    \frac{z^3(2z^3-2z^2+z+1)}{1-6z^2}=z^3 + z^4 + 4z^5 + 8z^6 + 24z^7 + 48z^8 + 144z^9 + \\ 288z^{10} + 864z^{11} + 1\,728z^{12} + 5\,184z^{13}+O(z^{14})
    \end{multline*}
    \begin{multline*}
    \frac{2z^4(2z^2-2z-1)}{1-6z^2}=2z^4 + 4z^5 + 8z^6 + 24z^7 + 48z^8 + 144z^9 + \\ 288z^{10} + 864z^{11} + 1\,728z^{12} + 5\,184z^{13}+O(z^{14})
    \end{multline*}
    \begin{multline*}
    \frac{z^3(2z^3-2z^2+z+1)}{1-6z^2}=z^3 + z^4 + 4z^5 + 8z^6 + 24z^7 + 48z^8 + 144z^9 + \\ 288z^{10} + 864z^{11} + 1\,728z^{12} + 5\,184z^{13}+O(z^{14})
    \end{multline*}
\item Removing the references to the tuck move, we recover generating functions for the number of windings available for each tie length. We give these for \texttt{R}-final, \texttt{L}-final and \texttt{C}-final respectively. Summed up, these simply enumerate all possible \texttt{T}/\texttt{W}-strings of the corresponding lengths, and so run through powers of 2.
\[
\frac{z^3}{1-z-2z^2} = z^3 + z^4 + 3z^5 + 5z^6 + 11z^7 + 21z^8 + 43z^9 + 85z^{10} + 171z^{11} + 341z^{12} + 683z^{13} + O(z^{14})
\]
\[
\frac{2z^4}{(1-2z)(1+z)} = 2z^4 + 2z^5 + 6z^6 + 10z^7 + 22z^8 + 42z^9 + 86z^{10} + 170z^{11} + 342z^{12} + 682z^{13} + O(z^{14})
\]
\[
\frac{z^3}{1-z-2z^2} = z^3 + z^4 + 3z^5 + 5z^6 + 11z^7 + 21z^8 + 43z^9 + 85z^{10} + 171z^{11} + 341z^{12} + 683z^{13} + O(z^{14})
\]
\item For the full grammar of arbitrary depth knots, we set $w$ to be a root of 
$(8z^6-4z^4)\zeta^3+(-8z^6+18z^4-7z^2)\zeta^2+(-16z^6+14z^4-6z^2+2)\zeta-12z^4+9z^2-2=0$ solved for $\zeta$.
Then the generating function for this grammar is:
\begin{multline*}
-\frac{1}{8z^4-11z^2+3}\large(64w^2z^7-128wz^7+32w^2z^6-64wz^6-48z^5w^2+216z^5w-24w^2z^4- \\
-96z^5+108wz^4+8w^2z^3-48z^4-110wz^3+4w^2z^2+82z^3-55z^2w+41z^2+16zw-16z+8w-8\large) = \\
2z^2+4z^3+20z^4+40z^5+192z^6+384z^7+1\,896z^8+3\,792z^9 + \\
19\,320 z^{10} + 38\,640 z^{11} + 202\,392 z^{12} + 404\,784 z^{13} + 2\,169\,784 z^{14} +
O(z^{15})
\end{multline*}
\end{itemize}

\subsection{Tables of counts}

For ease of reading, we extract the results from the generating functions above to more easy-to-reference tables here. Winding length throughout refers to the number of \texttt{R}/\texttt{L}/\texttt{C} symbols occurring, and thus is 1 larger than the \texttt{W}/\texttt{T} count.

The cases enumerated by \textcite{fink2000tie} are

\begin{center}
\begin{tabular}{l|rrrrrrr|l}
  Winding length & 2 & 3 & 4 & 5 & 6 & 7 & 8 & total  \\ \hline
  $\#$ tie-knots & 1 & 1 & 3 & 5 & 11 & 21 & 43 & 85
\end{tabular}
\end{center}


A knot with the thick blade active will cover up the entire knot with each new bow. As such, all thick blade active tie-knots will fall within the classification by \textcite{fink2000tie}.

The modern case, thus, deals with thin blade active knots. As evidenced by the Trinity and the Eldredge knots, thin blade knots have a wider range of interesting façades and of interesting tuck patterns. For thick blade knots, it was enough to assume that the tuck happens last, and from the C region, the thin blade knots have a far wider variety.

The case remains that unless the last move is a tuck -- or possibly finishes in the \texttt{C} region -- the knot will unravel from gravity. We can thus expect this to be a valid requirement for the enumeration. There are often more valid tuck sites than the final position in a knot, and the tuck need no longer come from the \texttt{C} region: \texttt{R} and \texttt{L} are at least as valid.

The computations in Section~\ref{sec:generatingfunctions} establish

\begin{center}
\begin{tabular}{l|rrrrrrrrrrr|r}
  Winding length & 3 & 4 & 5 & 6 & 7 & 8 & 9 & 10 & 11 & 12 & 13 & total\\ \hline
  \# left windings & 0 & 2 & 2 & 6 & 10 & 22 & 42 & 86 & 170 & 342 & 682 & 1\,364 \\
  \# right windings & 1 & 1 & 3 & 5 & 11 & 21 & 43 & 85 & 171 & 341 & 683 & 1\,365 \\
  \# center windings & 1 & 1 & 3 & 5 & 11 & 21 & 43 & 85 & 171 & 341 & 683 &  1\,365 \\ \hline
  \# left knots & 0 & 2 & 4 & 8 & 24 & 48 & 144 & 288 & 864 & 1\,728  & 5\,184 & 8\,294 \\
  \# right knots & 1 & 1 & 4 & 8 & 24 & 48 & 144 & 288 & 864 & 1\,728 & 5\,184 & 8\,294 \\
  \# center knots & 1 & 1 & 4 & 8 & 24 & 48 & 144 & 288 & 864 & 1\,728 & 5\,184 & 8\,294 \\ \hline
  \# single tuck knots & 2 & 4 & 12 & 24 & 72 & 144 & 432 & 864 & 2\,592 & 4\,146 
      & 15\,552 & 24\,882 \\  
  total \# knots & 2 & 4 & 20 & 40 & 192  &384 & 1\,896 & 3\,792 
      & 19\,320 & 38\,640 & 202\,392 & 266\,682
\end{tabular}
\end{center}

The first point where the singly tucked knots and the full range of knots deviate is at the knots with winding length 4; there are 12 singly tucked knots, and 8 knots that allow for a double tuck, 
namely: \\
\begin{tabular}{llll}
\texttt{TTTTU} & 
\texttt{TTWWU} & 
\texttt{TWTTU} & 
\texttt{TWWWU} \\
\texttt{WTTTU} & 
\texttt{WTWWU} & 
\texttt{WWTTU} & 
\texttt{WWWWU} \\
\texttt{TTUTTU} & 
\texttt{TTUWWU} & 
\texttt{WWUTTU} & 
\texttt{WWUWWU} \\
\texttt{TTTWUU} & 
\texttt{TTWTUU} & 
\texttt{TWTTUU} & 
\texttt{TWTTU'UU} \\
\texttt{WTWWUU} & 
\texttt{WTWWU'UU} & 
\texttt{WWTWUU} & 
\texttt{WWWTUU} 
\end{tabular}

The reason for the similarity between the right and the center counts is that the winding sequences can be mirrored. Left-directed knots are different since the direction corresponds to the starting direction.
Hence, a winding sequence for a center tuck can be mirrored to a winding sequence for a right tuck. 

In the preprint version of this paper, we claimed the total count of knots using only single-depth tucks to be 177\,147. During the revision of the paper, we have discovered two errors in this claim: 
\begin{enumerate}
\item There is an off-by-one error in this count.
\item This count was done for tie-knots that allow tucks that are hidden behind the knot. Adding this extra space to the generating grammar produces the generating function
\begin{multline*}
2z^3 + 6z^4 + 18z^5 + 54z^6 + 162z^7 + 486z^8 + 1\,458z^9 + \\
4\,374z^{10} + 13\,122z^{11} + 39\,366z^{12} + 118\,098z^{13} + O(z^{14})
\end{multline*}
with a total of 177\,146 tie-knots with up to 13 moves. 
\end{enumerate}

\section{Aesthetics}
\label{sec:aesthetics}

\textcite{fink2000tie} propose several measures to quantify the aesthetic qualities of a necktie-knot; notably \emph{symmetry} and \emph{balance}, corresponding to the quantities $\#R-\#L$ and the number of transitions from a streak of \texttt{W} to a streak of \texttt{T} or vice versa.

By considering the popular thin-blade neck tie-knots: the Eldredge and the Trinity, as described in \cite{eldredge-alex,trinity-alex}, we can immediately note that balance no longer seems to be as important for the look of a tie-knot as is the shape of its façade. Symmetry still plays an important role in knots, and is easy to calculate using the CLR notation for tie-knots.

\begin{center}
\begin{tabular}[t]{rllcc}
  Knot & \texttt{TW}-string & \texttt{CLR}-string & Balance & Symmetry \\ \hline
  \textbf{Eldredge} & \texttt{TTTWWTTUTTWWU} & \texttt{LCRLRCRLUCRCLU}   & 3 & 0 \\ 
  \textbf{Trinity} & \texttt{TWWWTTTUTTU}   & \texttt{LCLRCRLCURLU} & 2 & 1 \\
\end{tabular}
\end{center}

We do not in this paper attempt to optimize any numeric measures of aesthetics, as this would require us to have a formal and quantifiable measure of the knot façades. This seems difficult with our currently available tools.

\section{Conclusion}
\label{sec:conclusion}

In this paper, we have extended the enumeration methods originally used by \textcite{fink2000tie} to provide a larger enumeration of necktie-knots, including those knots tied with the thin blade of a necktie to produce ornate patterns in the knot façade.

We have found 4\,094 winding patterns that take up to 13 moves to tie and are anchored by a final single depth tuck, and thus are reasonable candidates for use with a normal necktie. We chose the number of moves by examining popular thin-blade tie-knots -- the Eldredge tie-knot uses 12 moves -- and by experimentation with our own neckties. Most of these winding patterns allow several possible tuck patterns, and thus the 4\,094 winding patterns give rise to 24\,882 singly tucked tie-knots.

We have further shown that in the limit, the language describing neck tie-knots is context free, with a regular sub-language describing these 24\,882 knots.

These counts, as well as the stated generating functions, are dependent on the correctness of the \texttt{combstruct} package in Maple, and the correctness of our encoding of these grammars as Maple code. We have checked the small counts and generated strings for each of the grammars against experiments with a necktie and with the results by Fink and Mao and our own catalogue. 

Questions that remain open include:
\begin{itemize}
\item Find a way to algorithmically divide a knot description string into a body/façade distinction.
\item Using such a distinction, classify all possible knot façades with reasonably short necktie lengths.
\end{itemize}

We have created a web-site that samples tie-knots from knots with at most 12 moves and displays tying instructions: \url{http://tieknots.johanssons.org}. The entire website has also been deposited with Figshare\cite{Vejdemo-Johansson2015}.

All the code we have used, as well as a table with assigned names for the 2\,046 winding patterns for up to 12 moves are provided as supplementary material to this paper. Winding pattern names start with \texttt{R}, \texttt{L} or \texttt{C} depending on the direction of the final tuck, and then an index number within this direction. We suggest augmenting this with the bit-pattern describing which internal tucks have been added -- so that e.g. the Eldredge would be L-373.4 (including only the 3rd potential tuck from the start) and the Trinity would be L-110.2 (including only the 2nd potential tuck). Thus, any single-depth tuck can be concisely addressed.

\section{Acknowledgements}

We would like to thank the reviewers, whose comments have gone a long way to make this a far better paper, and who have caught several errors that marred not only the presentation but also the content of this paper.

Reviewer 1 suggested a significant simplification of the full grammar in Section~\ref{sec:depthN}, which made the last generating function at all computable in reasonable time and memory.

Reviewer 2 suggested we look into generating functions as a method for enumerations. As can be seen in Section~\ref{sec:generatingfunctions}, this suggestion has vastly improved both the power and ease of most of the results and calculations we provide in the paper.

For these suggestions in particular and all other suggestions in general we are thankful to both reviewers.

\newpage
\printbibliography

\end{document}